\documentclass{article}
\usepackage{spconf,amsmath,amssymb,amsthm,graphicx,booktabs,url,balance,capt-of}
\usepackage{algorithm}
\usepackage{algpseudocode}
\usepackage{color}
\usepackage{xcolor}
\usepackage{acronym}
\newtheorem{theorem}{Theorem}

\newcommand{\Ind}{\mathbb{I}}
\newcommand{\herm}{^{\mathsf{H}}}

\newcommand{\Prob}{\mathbb{P}}

\newcommand{\vecx}[1]{\bm{#1}}
\newcommand{\matx}[1]{\bm{#1}}
\newcommand{\setx}[1]{\mathcal{#1}}

\newtheorem{proposition}{Proposition}

\acrodef{sinr}[SINR]{signal-to-interference-plus-noise ratio}
\acrodef{csi}[CSI]{channel state information}
\acrodef{crc}[CRC]{conformal risk control}
\acrodef{cp}[CP]{Conformal prediction}

\title{Conformal Risk Control for Multiuser Beamforming}

\name{Aman Goel and Sundeep Prabhakar Chepuri}
\address{Indian Institute of Science, Bengaluru, India}
\def\cast{{
   \mathord{
      \hbox to 0em{
         \ooalign{
	   \smash{\hbox{$\ast$}}\crcr
	   \smash{\hskip-1pt\Large\hbox{$\circ$}} }
	 \hidewidth}
      \phantom{\bigcirc}
} }}
\def\bm#1{\mbox{\boldmath $#1$}}

\newcommand{\bds}{\begin {itemize}}
\newcommand{\eds}{\end {itemize}}
\newcommand{\bdf}{\begin{definition}}
\newcommand{\blm}{\begin{lemma}}
\newcommand{\edf}{\end{definition}}
\newcommand{\elm}{\end{lemma}}
\newcommand{\bthm}{\begin{theorem}}
\newcommand{\ethm}{\end{theorem}}
\newcommand{\bprp}{\begin{prop}}
\newcommand{\eprp}{\end{prop}}
\newcommand{\bcl}{\begin{claim}}
\newcommand{\ecl}{\end{claim}}
\newcommand{\bcr}{\begin{coro}}
\newcommand{\ecr}{\end{coro}}
\newcommand{\bquest}{\begin{question}}
\newcommand{\equest}{\end{question}}

\newcommand{\larrow}{{\larrow}}

\newcommand{\argmin}{\ensuremath{\mathrm{arg}\min}}
\newcommand{\argmax}{\ensuremath{\mathrm{arg}\max}}

\usepackage{latexsym}

\def\IC{\mathbb C}
\def\IN{\mathbb N}
\def\IZ{\mathbb Z}
\def\IR{\mathbb R}

\def\shat{^{\mathchoice{}{}%
 {\,\,\smash{\hbox{\lower4pt\hbox{$\widehat{\null}$}}}}%
 {\,\smash{\hbox{\lower3pt\hbox{$\hat{\null}$}}}}}}

\def\bSigma{{
      \ooalign{
      \smash{\hskip.4pt\raise.4pt\hbox{$\Sigma$}}\vphantom{}\crcr
      \smash{\hskip.7pt\raise.6pt\hbox{$\Sigma$}}\vphantom{}\crcr
      \smash{\hbox{$\Sigma$}}\vphantom{$\Sigma$}}
      \vphantom{\hbox{$\Sigma$}}
      }}
\def\bTheta{{
      \ooalign{
      \smash{\hskip.5pt\raise.5pt\hbox{$\Theta$}}\vphantom{}\crcr
      \smash{\hskip.0pt\raise.1pt\hbox{$\Theta$}}\vphantom{}\crcr
      \smash{\hbox{$\Theta$}}\vphantom{$\Theta$}}
      \vphantom{\hbox{$\Theta$}}
      }}
\def\bDelta{{
      \ooalign{
      \smash{\hskip.4pt\raise.4pt\hbox{$\Delta$}}\vphantom{}\crcr
      \smash{\hskip.7pt\raise.6pt\hbox{$\Delta$}}\vphantom{}\crcr
      \smash{\hbox{$\Delta$}}\vphantom{$\Delta$}}
      \vphantom{\hbox{$\Delta$}}
      }}
\def\bLambda{{
      \ooalign{
      \smash{\hskip.5pt\raise.5pt\hbox{$\Lambda$}}\vphantom{}\crcr
      \smash{\hskip.0pt\raise.1pt\hbox{$\Lambda$}}\vphantom{}\crcr
      \smash{\hbox{$\Lambda$}}\vphantom{$\Lambda$}}
      \vphantom{\hbox{$\Lambda$}}
      }}

\makeatletter

\def\bordermatrix#1{\begingroup \m@th
  \@tempdima 8.75\p@
  \setbox\z@\vbox{%
    \def\cr{\crcr\noalign{\kern2\p@\global\let\cr\endline}}%
    \ialign{$##$\hfil\kern2\p@\kern\@tempdima&\thinspace\hfil$##$\hfil
      &&\quad\hfil$##$\hfil\crcr
      \omit\strut\hfil\crcr\noalign{\kern-\baselineskip}%
      #1\crcr\omit\strut\cr}}%
  \setbox\tw@\vbox{\unvcopy\z@\global\setbox\@ne\lastbox}%
  \setbox\tw@\hbox{\unhbox\@ne\unskip\global\setbox\@ne\lastbox}%
  \setbox\tw@\hbox{$\kern\wd\@ne\kern-\@tempdima\left[\kern-\wd\@ne
    \global\setbox\@ne\vbox{\box\@ne\kern2\p@}%
    \vcenter{\kern-\ht\@ne\unvbox\z@\kern-\baselineskip}\,\right]$}%
  \null\;\vbox{\kern\ht\@ne\box\tw@}\endgroup}
\makeatother

\makeatletter
\def\argmin{\mathop{\operator@font arg\,min}}
\def\argmax{\mathop{\operator@font arg\,max}}
\makeatother

\def\bm#1{\mbox{\boldmath $#1$}}

\newcommand{\bea}{\begin{array}}
\newcommand{\ena}{\end{array}}
\newcommand{\beq}{\begin{equation}}
\newcommand{\enq}{\end{equation}}

\newcommand{\beqa}{\begin{eqnarray}}
\newcommand{\enqa}{\end{eqnarray}}

\newcommand{\beqan}{\begin{eqnarray*}}
\newcommand{\enqan}{\end{eqnarray*}}

\newcommand{\AL}{\begin{enumerate}}
\newcommand{\ALE}{\end{enumerate}}

\def\addots{\mathinner{
    \mkern1mu\raise0pt\vbox{\kern7pt\hbox{.}}
    \mkern2mu\raise4pt\hbox{.}
    \mkern2mu\raise7pt\hbox{.}
    \mkern1mu}}

\def\sddots{\mathinner{
    \mkern.8mu\raise7pt\hbox{.}
    \mkern.8mu\raise4pt\hbox{.}
    \mkern.8mu\raise0pt\vbox{\kern7pt\hbox{.}}
    \mkern1mu}}

\def\saddots{\mathinner{
    \mkern.2mu\raise0pt\vbox{\kern7pt\hbox{.}}
    \mkern.2mu\raise4pt\hbox{.}
    \mkern.2mu\raise7pt\hbox{.}
    \mkern1mu}}

\newcommand{\bh}{{\ensuremath{\mathbf{h}}}}

\def\sqplus{\mathbin{
	{\ooalign{\hfil\raise.3ex\hbox{\scriptsize
	+}\hfil\crcr\mathhexbox274\crcr\mathhexbox275}}
	}} 
\def\sqminus{\mathbin{
	{\ooalign{\hfil\raise.3ex\hbox{\scriptsize
	--}\hfil\crcr\mathhexbox274\crcr\mathhexbox275}}
	}}

\def\IC{{
   \mathord{
      \hbox to 0em{
	 \hskip-4pt
         \ooalign{
	   \smash{\hskip1.9pt\raise2.6pt\hbox{$\scriptscriptstyle |$}}\crcr
	   \smash{\hbox{\rm\sf C}} }
	 \hidewidth}
      \phantom{\hbox{\rm\sf C}}
} }}
\def\IN{
    {\ooalign{
   \smash{\hskip2.2pt\raise1.5pt\hbox{$\scriptscriptstyle |$}}\vphantom{}\crcr
   \hbox{\sf N}
	}}
	} 
\def\IZ{
    {\ooalign{
   \smash{\hskip1.9pt\raise0pt\hbox{$\sf Z$}}\vphantom{}\crcr
   \hbox{\sf Z}
	}}
	} 
\def\IR{
    {\ooalign{
   \smash{\hskip2.2pt\raise1.5pt\hbox{$\scriptscriptstyle |$}}\vphantom{}\crcr
   \smash{\hskip2.2pt\raise3.3pt\hbox{$\scriptscriptstyle |$}}\vphantom{}\crcr
   \hbox{\sf R}
	}}
	} 

\DeclareMathAlphabet{\mathcmb}{OT1}{cmr}{b}{n}

\def\bSigma{\ensuremath{\mathcmb{\Sigma}}}
\def\bLambda{\ensuremath{\mathcmb{\Lambda}}}

\def\bTheta{\ensuremath{\mathcmb{\Theta}}}

\newcommand{\SI}{\begin{indlist}}
\newcommand{\EI}{\end{indlist}}

\newcommand{\DL}{\begin{dashlist}}
\newcommand{\DLE}{\end{dashlist}}

\makeatletter
\def\setboxz@h{\setbox\z@\hbox}
\def\wdz@{\wd\z@}
\def\boxz@{\box\z@}
\def\underset#1#2{\binrel@{#2}%
  \binrel@@{\mathop{\kern\z@#2}\limits_{#1}}}
\def\binrel@#1{\begingroup
  \setboxz@h{\thinmuskip0mu
    \medmuskip\m@ne mu\thickmuskip\@ne mu
    \setbox\tw@\hbox{$#1\m@th$}\kern-\wd\tw@
    ${}#1{}\m@th$}%
  \edef\@tempa{\endgroup\let\noexpand\binrel@@
    \ifdim\wdz@<\z@ \mathbin
    \else\ifdim\wdz@>\z@ \mathrel
    \else \relax\fi\fi}%
  \@tempa
}
\let\binrel@@\relax%
\makeatother

\begin{document}
\ninept
\maketitle

\begin{abstract}
We consider multiuser downlink beamforming with prescribed
signal-to-interference-plus-noise ratio (SINR) requirements under imperfect channel state information and a tight transmit-power budget. Robust worst-case designs often require tractable channel-error models and may not yield feasible beamformers that meet the power budget. To alleviate these issues, we propose conformal risk control for beamforming (CRCBeam), which combines the classical fixed-point minimum-power beamforming with conformal risk control to calibrate an SINR margin from calibration sample pairs of estimated and true channels. Under exchangeability, CRCBeam provides finite-sample control over served-outage risk, which is the probability of transmitting when at least one user fails to meet its SINR target. CRCBeam requires no parametric channel-error model and transmits only when the power budget is met. Numerical experiments using a pilot contamination channel error model corroborate that CRCBeam achieves the served-outage risk while meeting the power budget. In contrast, worst-case robust baselines, Gaussian-ball and conformal-ball beamforming, either do not meet the target risk or exceed the power budget.
\end{abstract}

\vspace{-2mm}
\section{Introduction}~\label{sec:intro}
Multiuser downlink beamforming seeks to meet users' \ac{sinr}
requirements with minimum transmit power. With perfect \ac{csi},
this classical problem admits efficient fixed-point algorithms~\cite{bjornson2014}. In practice, the base
station has imperfect channel estimates, so satisfying
estimated-channel SINR constraints does not ensure that users'
quality-of-service requirements are met on the true channels.
Moreover, the joint distribution of true and estimated channels can
be difficult to characterize, particularly with learned channel
estimators, hardware impairments, or intermittent pilot contamination.
This motivates calibrating reliability directly from data without a parametric CSI-error model. Practical operation also imposes a hard transmit power budget, i.e., a candidate beamformer is deemed as deployable only if its required power does not exceed this budget.

Worst-case robust methods enforce SINR requirements over a prescribed
range of possible channels~\cite{vorobyov2005robust}, while model-based
chance-constrained methods limit outage under an assumed channel-error
distribution~\cite{vorobyov2005robustchance}. Their reliability depends on how
accurately these descriptions represent actual CSI errors.
\ac{cp} is a popular machine learning technique that offers a data-driven alternative by constructing channel sets
with prescribed marginal coverage under
exchangeability~\cite{angelopoulos2021gentle}.
Such sets have been combined with robust beamforming to optimize a
single user's rate under a power constraint~\cite{su2025conformal}.
For prescribed multiuser SINR targets, however, satisfying the
requirements for every channel in a calibrated set can demand more
power than the permitted budget, thereby reducing transmission availability.

We propose \emph{conformal risk control for beamforming} (CRCBeam),
which applies \ac{crc}~\cite{angelopoulos2024crc} to control the probability of transmitting while at
least one user fails its SINR target, which we call  as \emph{served-outage risk}. Using pairs of estimated and true
channels, it calibrates a scalar SINR margin for nominal minimum-power
beamforming without constructing a channel uncertainty set.
The resulting policy provides finite-sample marginal risk control
under exchangeability while enforcing the power budget.
Our contributions are as follows.
\begin{itemize}\itemsep2pt
\item We develop CRCBeam that uses a fixed-point beamforming solver and provides served-outage control.
\item We establish finite-sample marginal served-outage control using
a monotone upper envelope of the proposed risk function. We also show that the selected regular
margin has the lowest optimal required power among all CRC-certified
regular margins on the candidate grid.
\end{itemize}
We conduct experiments using an intermittent pilot contamination model to synthesize channel errors and evaluate served-outage and transmit power. CRCBeam achieves low served-outage and meets the transmit power budget, whereas robust baselines that construct uncertainty sets based on a Gaussian assumption or CP either fail to meet the target risk or exceed the power budget.

\vspace{-2mm}
\section{Problem Modelling}
\label{sec:problem}

\noindent\textbf{System model.}
Consider a downlink system with an $M$-antenna base station serving
$K$ single antenna users. Let $s_k\in\mathbb{C}$ denote the data symbol
of user $k$, and let $\vecx h_k\in\mathbb{C}^{M}$ denote its channel
vector. The symbols satisfy $\mathbb E[s_j s_k^*]=0$ for $j\neq k$
and $\mathbb E[|s_k|^2]=1$. They are linearly precoded using
$\vecx w_1,\ldots,\vecx w_K\in\mathbb{C}^{M}$. The received signal is
$r_k=\vecx h_k\herm\bigl(\sum_{j=1}^{K}\vecx w_j s_j\bigr)+n_k$,
where $n_k$ is zero-mean receiver noise with variance $\sigma^2>0$,
independent of the symbols. Writing
$\matx W=[\vecx w_1,\ldots,\vecx w_K]\in\mathbb{C}^{M\times K}$,
the required transmit power is
$\|\matx W\|_F^2=\sum_{k=1}^{K}\|\vecx w_k\|_2^2$.
The \ac{sinr} at user $k$ is
\begin{equation}
\mathrm{SINR}_k(\matx W,\vecx h_k)
=
\frac{|\vecx h_k\herm\vecx w_k|^2}
{\displaystyle\sum_{j\neq k}|\vecx h_k\herm\vecx w_j|^2+\sigma^2}.
\label{eq:sinr}
\end{equation}
Under perfect \ac{csi}, the classical minimum-power beamformimg problem
is~\cite{bjornson2014}
\begin{align}
\underset{\matx W}{\operatorname{minimize}}
\quad & \|\matx W\|_F^2
\tag{P1a}\label{prob:p1a}\\
\text{subject to}
\quad & \mathrm{SINR}_k(\matx W,\vecx h_k)\geq\gamma_k,
\quad k=1,\ldots,K,
\tag{P1b}\label{prob:p1b}
\end{align}
where $\gamma_k>0$ is user $k$'s SINR target. A candidate is
\emph{deployable} only if $\|\matx W\|_F^2\leq P_{\max}$, where
$P_{\max}>0$ is the power budget. Over-budget candidates are not
transmitted, and no user is served in that transmission opportunity.

In practice, only imperfect channel estimates
$\{\widehat{\vecx h}_k\}_{k=1}^{K}$ are available. Let
$\matx H=[\vecx h_1,\ldots,\vecx h_K]$ and
$\widehat{\matx H}=[\widehat{\vecx h}_1,\ldots,\widehat{\vecx h}_K]$.
Let $Q$ denote the unknown joint distribution of
$(\matx H,\widehat{\matx H})$. We do not assume a parametric form
for this distribution. A nominal design substitutes
$\widehat{\vecx h}_k$ for $\vecx h_k$ in \eqref{prob:p1b}.
However, $\mathrm{SINR}_k(\matx W,\widehat{\vecx h}_k)\geq\gamma_k$
does not imply $\mathrm{SINR}_k(\matx W,\vecx h_k)\geq\gamma_k$.
Nominal feasibility therefore does not ensure true-channel reliability.

\noindent\textbf{Problem statement.}
This paper develops a conformal risk control (CRC) framework for
beamforming with imperfect \ac{csi} and power-capped deployment.
Our objective is to minimize average required transmit power while
controlling the probability that a candidate is deployed and at least
one user fails to meet its SINR target. We call this event a
\emph{served outage}. No transmission outcomes are excluded from this event, and availability is assessed separately.

Let $\matx W(\widehat{\matx H})$ denote the candidate beamformer computed from the available channel estimate. The candidate may be produced by an optimization-based or learned neural beamforming
method. Define its availability
indicator and joint outage loss, respectively, as
\begin{equation}
p(\matx W)=\Ind\!\left\{\|\matx W\|_F^2\leq P_{\max}\right\},
\label{eq:deployment-indicator}
\end{equation}
\begin{equation}
\ell(\matx W,\matx H)
=
\Ind\!\left\{
\exists\,k\in\{1,\ldots,K\}:\,
\mathrm{SINR}_k(\matx W,\vecx h_k)<\gamma_k
\right\},
\label{eq:joint-outage-loss}
\end{equation}
where $\Ind\{\cdot\}$ denotes the indicator function.
The joint served-outage loss is
\begin{equation}
\setx S(\matx W,\matx H)=p(\matx W)\ell(\matx W,\matx H).
\label{eq:served-outage-loss}
\end{equation}
It equals one only when a transmission occurs and at least one user
is in outage. It equals zero for successful transmission
($p=1$, $\ell=0$) or no transmission ($p=0$). Zero served-outage loss therefore does not necessarily indicate successful service.

For a prescribed risk level $\alpha\in(0,1)$, the proposed population design problem\footnote{Unless stated otherwise, probabilities are taken with respect to the joint distribution $Q$ of $(\widehat{\matx H},\matx H)$ and expectation
subscripts indicate the variables being averaged over with
$\mathbb E_{\widehat{\matx H}}$ using the corresponding marginal
distribution and $\mathbb E_{\widehat{\matx H},\matx H}$ using the
joint distribution.} is
\begin{equation}
\begin{aligned}
\underset{\matx W(\cdot)}{\operatorname{minimize}}
\quad&
\mathbb E_{{\widehat{\matx H}}}
\left[\|\matx W(\widehat{\matx H})\|_F^2\right]
\\
\text{subject to}\quad&
\mathbb E_{{\widehat{\matx H},\matx H}}
\left[
\setx S\bigl(\matx W(\widehat{\matx H}),\matx H\bigr)
\right]\leq\alpha.
\end{aligned}
\tag{P2}
\label{eq:beamforming-problem}
\end{equation}
The optimization is over beamforming rules $\matx W(\cdot)$ based on
$\widehat{\matx H}$ with finite average required power. The objective
measures candidate power, including over-budget candidates, rather
than actual power transmitted after applying the budget check.
The risk constraint is equivalent to
\begin{equation}
\Prob\!\left\{
\begin{aligned}
&\|\matx W(\widehat{\matx H})\|_F^2\leq P_{\max},\\
&\exists\,k:\,
\mathrm{SINR}_k\bigl(\matx W(\widehat{\matx H}),\vecx h_k\bigr)
<\gamma_k
\end{aligned}
\right\}\leq\alpha.
\label{eq:served_outage_probability}
\end{equation}
This is a marginal served-outage constraint, which limits the probability of an actual transmission failing at least one user's SINR requirement, while accommodating no transmission decisions imposed by the power budget. We develop a CRC-based approach for enforcing the served-outage
constraint in \textup{(P2)} using finite calibration data.
Define the deployment probability as
$p_{\mathrm{dep}}
=\mathbb E_{\widehat{\matx H}}
[p(\matx W(\widehat{\matx H}))]$.
Any design satisfying \textup{(P2)} serves all users successfully
with probability
$p_{\mathrm{dep}}-\mathbb E_{\widehat{\matx H},\matx H}[\setx S]
\geq p_{\mathrm{dep}}-\alpha$.
Thus, availability and served-outage risk jointly characterize the probability of successful service.

\vspace{-2mm}
\section{Prior Art and Background}
\label{sec:prelim}
We review how existing approaches account for channel uncertainty
and motivate direct calibration of the served-outage loss.
\\

\noindent\textbf{Robust Beamforming under Imperfect CSI.}
Worst-case designs prescribe sets such as
$\mathcal H_k(\widehat{\vecx h}_k)
=\{\vecx h:\|\vecx h-\widehat{\vecx h}_k\|_2\leq\varepsilon_k\}$,
where $\varepsilon_k>0$, and replace \eqref{prob:p1b}
with~\cite{vorobyov2005robust}
\begin{equation}
\mathrm{SINR}_k(\matx W,\vecx h)\geq\gamma_k,
\quad\forall\vecx h\in\mathcal H_k(\widehat{\vecx h}_k).
\label{eq:robust_constraint}
\end{equation}
The so-called S-procedure reformulates these quadratic uncertainty
constraints~\cite{boyd2004convex}, but does not validate the chosen
sets. Small sets miss relevant errors, whereas large sets can require
excessive power. Model-based chance-constrained designs instead
assume a conditional channel law, often Gaussian{~\cite{vorobyov2005robustchance}.
Their outage bounds need not hold when this law differs from
$P_{\matx H\mid\widehat{\matx H}}$. These approaches require a
justified uncertainty set or probabilistic model, respectively.
\\

\noindent\textbf{Conformal prediction (CP).}
CP calibrates coverage without a parametric distributional model, under exchangeability of calibration and test pairs~\cite{angelopoulos2021gentle}. Let
$Z_i=(\widehat{\matx H}_i,\matx H_i)$ for $i=1,\cdots,N_c$ be $N_c$ calibration pairs
and a test pair $Z_{N_c+1}$. Assume exchangeability conditional on
separate training data, i.e., their joint distribution is invariant to reordering. For a non-conformity score $s_k$ and coverage level $\alpha_k \in (0,1)$ fixed before calibration, let $\widehat q_k$ be the
$\lceil(N_c+1)(1-\alpha_k)\rceil$-th smallest element of the $N_c$
calibration scores augmented by $+\infty$, retaining multiplicities. Then
\begin{equation}
\mathcal H_k^{\mathrm{CP}}(\widehat{\vecx h}_k)
=\{\vecx h:s_k(\widehat{\vecx h}_k,\vecx h)\leq\widehat q_k\}.
\label{eq:cp_set}
\end{equation}
Let $C_k = \{\vecx h_k \in \mathcal{H}_k^{\mathrm{CP}}(\widehat{\bh}_k)\}$ denote the coverage event for user $k$.
Under exchangeability of calibration and test pairs and 
$\sum_k\alpha_k\leq\alpha$, CP guarantees 
\begin{equation}
\begin{aligned}
\Prob(C_k)&\geq1-\alpha_k,\quad 
\Prob\left(\bigcap_{k=1}^{K}C_k\right)
\geq1-\sum_{k=1}^{K}\alpha_k\geq1-\alpha.
\end{aligned}
\label{eq:cp_coverage}
\end{equation}
CP-robust baselines use these sets in
\eqref{eq:robust_constraint}~\cite{su2025conformal}. The guarantee in \eqref{eq:cp_coverage}, however, is only a
\emph{channel-coverage} guarantee. If robust-feasible and within-budget candidates are deployed, this also guarantees joint served-outage risk at most $\alpha$. However, since CP protects every included channel, it can be very conservative. That is, channel miscoverage need not necessarily cause an SINR outage, while larger sets can increase power and reduce availability. 
\\

\noindent\textbf{Conformal risk control (CRC).}
Conformal risk control (CRC) extends the idea of CP from coverage of a prediction set to control of a general bounded risk function~\cite{angelopoulos2024crc}.
Fix a decision family and a finite grid
$\lambda_1<\cdots<\lambda_{J+1}$ before calibration, and let
$L_i(\lambda)=L(\lambda;Z_i)\in[0,1]$.
Suppose the calibration and test losses are exchangeable,
non-increasing in the grid index, and satisfy
$L_i(\lambda_{J+1})=0$ almost surely. Classical CRC selects
\begin{equation}
\widehat j=\min\left\{
 j:\,
\frac{1+\sum_{i=1}^{N_c}L_i(\lambda_j)}{N_c+1}\leq\alpha
\right\},
\label{eq:crc_background}
\end{equation}
over $j\in\{1,\ldots,J+1\}$, using $\widehat j=J+1$ if the set is
empty. It guarantees
$\mathbb E[L_{N_c+1}(\lambda_{\widehat j})]\leq\alpha$, averaging over
calibration and test data. CP controls set miscoverage loss, whereas CRC controls downstream decision losses. This distinction is particularly useful for the beamforming problem of interest. In essence, rather than first calibrating a channel uncertainty set and then protecting against every channel in that set, we can calibrate directly with respect to SINR failure on the true channels as discussed next.

%%%%%%%%%%%%%%%%%%%%%%%%%
%Proposed section
%%%%%%%%%%%%%%%%%%%%%%%%%%
\vspace*{-2mm}
\section{The Proposed {CRCBeam} Method}
\label{sec:method}

The proposed method addresses \eqref{eq:beamforming-problem} using finite calibration data to control the joint served-outage risk and select the minimum-power solution among CRC-certified candidates within a family of nominal SINR-margin beamformers.
In particular, it selects a scalar SINR margin and
applies the same power-budget check during calibration and deployment,
without specifying a channel-error distribution or uncertainty set.
Let $Z_i=(\widehat{\matx H}_i,\matx H_i)$ contain the estimated and
true channels of all users. The $N_c$ calibration pairs and a test
pair are assumed exchangeable, conditional on any separate training
data. 
\\

\noindent\textbf{Margin Calibration of CRCBeam.}
%\label{subsec:direct_crc}
Fix nonnegative ordered margins
$0=\lambda_1<\cdots<\lambda_J$.
For estimates $\widehat{\matx H}$ and margin $\lambda_j$, solve
\begin{equation}
\begin{aligned}
&\underset{\matx W}{\operatorname{minimize}}
\quad\|\matx W\|_F^2\\
&\text{s. to}\quad
\mathrm{SINR}_k(\matx W,\widehat{\vecx h}_k)
\geq\gamma_k(1+\lambda_j),\quad  k=1,\ldots,K.
\end{aligned}
\label{eq:margin_bf}
\end{equation}
This preserves the estimated CSI minimum-power formulation, i.e.,
$\lambda_1=0$ gives the nominal design, and larger margins tighten its SINR targets. 

Let $\matx W_i(\lambda_j)$ be the candidate returned for calibration sample~$i$,
and set $P_i(\lambda_j)=\|\matx W_i(\lambda_j)\|_F^2$.
If no valid candidate is returned, set $P_i(\lambda_j)=+\infty$.
For a valid candidate, define
$\ell_i(\lambda_j)=\ell(\matx W_i(\lambda_j),\matx H_i)$ using
\eqref{eq:joint-outage-loss}, evaluated at the original targets~$\gamma_k$; otherwise set $\ell_i(\lambda_j)=0$ by convention.
The availability and served-outage indicators, respectively, are
\begin{equation*}
p_i(\lambda_j)=\Ind\!\left\{P_i(\lambda_j)\leq P_{\max}\right\} \quad \text{and} \quad 
\setx S_i(\lambda_j)=p_i(\lambda_j)\ell_i(\lambda_j).
\label{eq:served_outage_loss}
\end{equation*}
Thus, unavailable candidates contribute zero served-outage loss and remain in the calibration sample count.

Increasing the margin tightens the estimated-channel SINR constraints, but need not improve SINR on the true channels. Hence,
$\setx S_i(\lambda_j)$ need not be non-increasing in $j$. We therefore construct a non-increasing upper envelope for CRC~\cite{angelopoulos2024crc}.
Append a no-transmission action $d$,
with $p_i(d)=\setx S_i(d)=0$, and write
$a_j=\lambda_j$ for $j\leq J$ and $a_{J+1}=d$. Define
\begin{equation}
\widetilde{\setx S}_i(a_j)
=\max_{m\in\{j,\ldots,J+1\}}\setx S_i(a_m).
\label{eq:monotone_envelope}
\end{equation}
The envelope is non-increasing in $j$, upper-bounds $\setx S_i$, and has terminal value zero. Define
\begin{equation}
\mathcal J_{\mathrm{CRC}}
=\left\{j:\,
\frac{1+\sum_{i=1}^{N_c}\widetilde{\setx S}_i(a_j)}{N_c+1}
\leq\alpha\right\},
\label{eq:crc_accepted_indices}
\end{equation}
where $j\in\{1,\ldots,J+1\}$. Then CRC selects
\begin{equation}
\widehat j=
\begin{cases}
\min\mathcal J_{\mathrm{CRC}},&\mathcal J_{\mathrm{CRC}}\neq\varnothing,\\
J+1,&\mathcal J_{\mathrm{CRC}}=\varnothing.
\end{cases}
\label{eq:crc_margin_selection}
\end{equation}
This returns the smallest regular margin meeting the CRC criterion
whenever one exists, and otherwise returns $d$. The explicit fallback
also covers risk levels below $1/(N_c+1)$.

\noindent\textbf{Power-Constrained Deployment.} If $\widehat j\leq J$, fix $\widehat\lambda=\lambda_{\widehat j}$
and solve \eqref{eq:margin_bf} for each new channel estimate. Deploy only a valid candidate with required power at most $P_{\max}$;
otherwise, do not transmit. If $d$ is selected, no transmission occurs.
Using the same solver, selected margin, and acceptance checks preserves
the calibrated decision rule. The next section establishes its
finite-sample marginal served-outage guarantee and the least-power
property among CRC-certified regular margins when the candidate
problems are solved globally.
\\

\noindent\textbf{Theoretical Guarantees.} The proposed method uses finite calibration data to control the joint
served-outage risk in \eqref{eq:beamforming-problem}.
Theorem~\ref{thm:crc} establishes this guarantee for a future deployment sample. Proposition~\ref{thm:least_power} shows that the selected regular
margin requires the least transmit power among all regular margins
satisfying the CRC criterion, provided each candidate beamforming
problem is solved to global optimality.

\begin{theorem}[Finite-sample joint served-outage control]
\label{thm:crc}
Let the $N_c$ calibration samples and a future deployment sample be exchangeable. Fix the margin grid and a measurable candidate generation
and deployment rule before calibration, and apply the same rule to every sample. Let $a_{\widehat j}$ be selected by~\eqref{eq:crc_margin_selection}, including its no-transmission fallback $d$. Then, for any $\alpha\in(0,1)$
\begin{equation}
\mathbb E\!\left[
\setx S_{N_c+1}(a_{\widehat j})
\right]\leq\alpha,
\label{eq:crc_guarantee_main}
\end{equation}
where the expectation is over both the calibration samples and the future deployment sample.
\end{theorem}

\noindent\emph{Proof.}
By \eqref{eq:monotone_envelope},
$0\leq\setx S_i(a_j)\leq\widetilde{\setx S}_i(a_j)\leq1$, and the envelope is non-increasing in $j$ and satisfies
$\widetilde{\setx S}_i(d)=0$.
For $\alpha\geq1/(N_c+1)$, the terminal action satisfies the CRC
criterion. The finite-grid CRC guarantee for bounded monotone
losses~\cite{angelopoulos2024crc} therefore gives
\[
\mathbb E\!\left[
\widetilde{\setx S}_{N_c+1}(a_{\widehat j})
\right]\leq\alpha.
\]
Since $\mathcal{S}_{N_c+1}(a_{\widehat\jmath})\leq\widetilde{\mathcal{S}}_{N_c+1}(a_{\widehat\jmath})$, we obtain (14).
For $\alpha<1/(N_c+1)$, the acceptance set is empty, so the selection
rule returns $d$, whose served-outage loss is identically zero.
\hfill$\square$

\begin{figure*}[t]
    \centering
\includegraphics[width=0.9\linewidth]{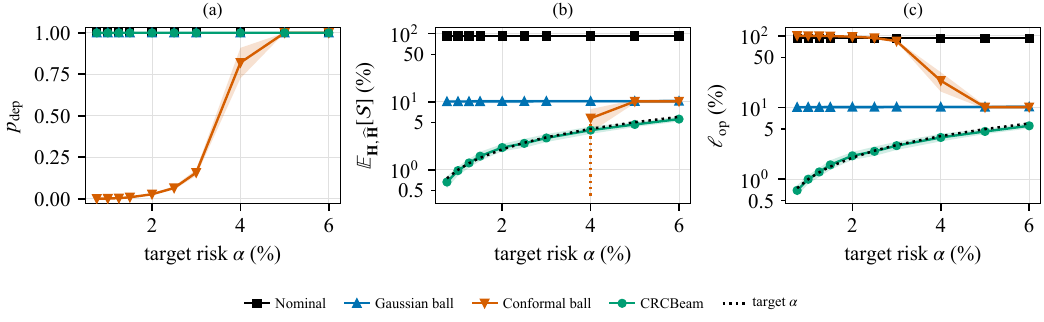}
    \caption{\footnotesize Performance versus $\alpha$ at $p_c=4\%$: (a) availability
    $p_{\rm dep}$; (b) joint served-outage $E_{\matx H, \widehat{\matx H}}[\setx S]$, with reference
    $E_{\matx H, \widehat{\matx H}}[\setx S]=\alpha$; (c) operational loss
    $\ell_{\rm op}$.
    Curves are means over $20$ replicates; bands are $95\%$ bootstrap
    intervals obtained by resampling replicates.}
    \label{fig:main}
    \vspace{-4mm}
\end{figure*}
If a beamforming mapping is learned from a separate training dataset
$\mathcal D_{\rm train}$, e.g., using a neural network, the learned
mapping is held fixed during calibration and deployment. If the calibration and deployment samples are exchangeable conditional
on $\mathcal D_{\rm train}$, the preceding proof yields
$\mathbb E[\setx S_{N_c+1}(a_{\widehat j})
\mid \mathcal D_{\rm train}]\leq\alpha$
almost surely, where the expectation averages over calibration and
deployment samples. Averaging over $\mathcal D_{\rm train}$ gives \eqref{eq:crc_guarantee_main}.

\begin{proposition}[Least-power CRC-certified margin]
\label{thm:least_power}
For fixed channel estimates $\widehat{\matx H}$, let
$P^\star(\lambda_j;\widehat{\matx H})$ denote the optimal value of
\eqref{eq:margin_bf}, with value $+\infty$ if the problem is
infeasible. This optimal value is non-decreasing in $j$.
Consequently, whenever $\widehat j\leq J$, for every regular margin
$\lambda_j$ satisfying the CRC criterion
\begin{equation}
P^\star(\lambda_{\widehat j};\widehat{\matx H})
\leq P^\star(\lambda_j;\widehat{\matx H}).
\label{eq:least_power_crc}
\end{equation}
Thus, a globally optimal candidate at the selected margin has the
minimum required transmit power within this certified family,
whenever that candidate problem is feasible.
\end{proposition}

\noindent\emph{Proof.}
For $j<m$, $\gamma_k(1+\lambda_j)\leq\gamma_k(1+\lambda_m)$
for every user $k$. Hence, the feasible set at $\lambda_m$ is
contained in that at $\lambda_j$. Minimizing the same power objective
over these nested sets gives
$P^\star(\lambda_j;\widehat{\matx H})
\leq P^\star(\lambda_m;\widehat{\matx H})$,
including infeasible cases under the stated convention.
The selection rule chooses the smallest certified index whenever a
regular margin is certified. Every other certified regular margin
therefore has index at least $\widehat j$, establishing
\eqref{eq:least_power_crc}.
\hfill$\square$

\vspace*{-2mm}
\section[Numerical Results]{Numerical Results\protect\footnotemark}
\label{sec:experiments}
\footnotetext{Code repository: {\scriptsize \url{github.com/amangoel-iisc/crcbeam}}}

We compare the proposed CRCBeam method with nominal, Gaussian-ball, and conformal-ball beamforming under a common power budget.
\\

\noindent\textbf{Setup.} An $M=8$-antenna base station serves $K=3$ users at angles
$-25^\circ$, $0^\circ$, and $25^\circ$, with unit large-scale gains,
Rician factor $6$~dB, SINR targets $\gamma_k=5$~dB, and
$\sigma^2=0.1$. To model intermittent pilot
contamination, we generate~\cite{figueiredo2018pilot,carvalho2016random}
\begin{equation}
\widehat{\vecx h}_k=\vecx h_k+\mathbf n_k+
C_k\sqrt{\rho_c}\,\mathbf g_k^{(c)},
\quad C_k\sim\mathrm{Bernoulli}(p_c),
\label{eq:sim_csi}
\end{equation}
where $\mathbf n_k\sim\mathcal{CN}(\mathbf0,
\mathbf I_M/(\tau_p\rho_p))$, $\tau_p=16$, and $\rho_p=25$~dB.
The independent co-pilot channel $\mathbf g_k^{(c)}$ has Rician factor
$3$~dB and angle uniform on $[-70^\circ,70^\circ]$.
We fix $\rho_c=-7.5$~dB. Neither $p_c$ nor $\rho_c$ is provided to CRCBeam. 

Each of $20$ independent replicates uses $3000$ calibration and $400$
test samples with margin grid $\{0,0.05,\ldots,2.0\}$.
A separate $4000$ sample development set determines $P_{\rm ref}$ as the $95$th percentile of nominal required power, and we set $P_{\max}=4P_{\rm ref}$.
Gaussian-ball models only thermal noise pilot estimation error, whereas CP calibrates its
radius from observed residuals. Both use the same S-procedure-based
solver. All methods transmit only valid candidates within the budget. Reported
uncertainty bands are $95\%$ bootstrap intervals over replicate-level means. We report deployment probability $p_{\mathrm{dep}}$ in Fig.~\ref{fig:main}(a), served-outage
risk $\mathbb E_{\matx H, \widehat{\matx H}}[\setx S]$ in Fig.~\ref{fig:main}(b), and operational failure probability
$\ell_{\mathrm{op}}=1-p_{\mathrm{dep}}+\mathbb E_{\matx H, \widehat{\matx H}}[\setx S]$ in Fig.~\ref{fig:main}(c). The latter counts both non-transmission and transmission with
at least one user in outage.
\\

\noindent\textbf{Results.}
At $p_c=4\%$, CRCBeam tracks the requested risk with nearly full
availability as shown in Fig.~\ref{fig:main}. At $\alpha=1\%$, its served-outage
is $0.98\%$ and availability is $99.98\%$.
Nominal beamforming remains deployable but suffers severe true-channel
outage, illustrating the consequence of designing directly on imperfect
CSI. The two uncertainty-set baselines exhibit complementary failure modes. Gaussian-ball remains fully available but has approximately $10\%$ served-outage across stringent
targets. Thus, tightening
the reliability parameter of the assumed Gaussian error model produces
little corresponding improvement in downstream reliability once the
actual CSI errors contain intermittent contamination. The conformal-ball
design behaves oppositely. Conformal-ball achieves near-zero empirical served-outage but only
$0.19\%$ availability at $\alpha=1\%$. Consequently,
Fig.~\ref{fig:main}(c) shows why served-outage and availability must be
considered jointly: a small served-outage risk obtained by rarely transmitting does not imply low operational loss.

Table~\ref{tab:pc_sweep} examines contamination sensitivity with
independent recalibration at each $p_c$. CRCBeam maintains mean
served-outage below $1\%$ and availability above $99.9\%$ throughout,
while its mean deployed power increases with contamination. As contamination becomes more frequent, Gaussian-ball design
continues to deploy on every test state, but its realized downstream risk increases. In contrast,
the conformal-ball design maintains essentially zero empirical served-outage risk, while its availability drops significantly. Together, Fig.~\ref{fig:main} and Table~\ref{tab:pc_sweep} reveal an empirical mismatch-induced outage floor for the Gaussian design. On the other hand, {CRCBeam} instead adapts the beamforming margin to the observed downstream losses, trading additional transmit power for controlled served-outage while preserving deployability.

\begin{table}[t]
\centering
\caption{\footnotesize Contamination sensitivity at $\alpha=1\%$. Entries are
$(E_{\matx H, \widehat{\matx H}}[\setx S],p_{\rm dep})$ in percent. The final column gives
(CRCBeam, Gaussian) mean deployed power in dB relative to $P_{\rm ref}$.
Values are means over $20$ independently recalibrated replicates.}
\label{tab:pc_sweep}
\scriptsize
\setlength{\tabcolsep}{2.8pt}
\renewcommand{\arraystretch}{1.05}
\begin{tabular}{@{}ccccc@{}}
\toprule
$p_c$ & CRCBeam & Gaussian & CP & Power (dB)\\
\midrule
$0\%$ & $(0.04,100.00)$ & $(0.00,100.00)$ & $(0.00,100.00)$
 & $(-0.63,-0.41)$\\
$2\%$ & $(0.91,100.00)$ & $(4.94,100.00)$ & $(0.00,2.58)$
 & $(2.03,-0.42)$\\
$4\%$ & $(0.98,99.98)$ & $(10.10,100.00)$ & $(0.00,0.19)$
 & $(2.69,-0.44)$\\
$6\%$ & $(0.76,99.95)$ & $(14.94,100.00)$ & $(0.00,0.05)$
 & $(3.08,-0.42)$\\
\bottomrule
\end{tabular}
\end{table}

\vspace*{-2mm}
\section{Conclusion and Future Work}
\label{sec:conclusion}

We proposed CRCBeam, a power-constrained multiuser beamforming
policy that calibrates served-outage risk without constructing channel uncertainty sets. Under exchangeability, it provides finite-sample risk control without a parametric CSI-error model and selects the certified regular margin with the lowest optimal required power on the candidate grid. Numerical experiments on a synthetic channel error model based on intermittent pilot contamination show low served-outage and high availability, supported by increased transmit power. In comparison, Gaussian-ball beamforming exceeds the target risk, while conformal-ball beamforming loses availability at 
stringent targets. These results demonstrate the practical
reliability--power tradeoff enabled by direct risk calibration.
Future work will address non-exchangeable channel evolution and
larger multiuser systems.

\clearpage
\section{Acknowledgements}
This work is supported by the IISc grant and the MHRD fellowship from the Government of India. Beyond this support, the authors have no relevant financial or non-financial interests to disclose. Claude AI was used to assist in writing the manuscript and coding the simulations. The authors take full responsibility for the results of this paper.

\section{Compliance with Ethical Standards}

This is a numerical simulation study for which no ethical approval was required.
\balance
\bibliographystyle{IEEEbib}
\bibliography{references_crc}
\end{document}